\documentclass[11pt]{llncs}
\usepackage[top=2cm, bottom=2cm, right=2.5cm, left=2.5cm]{geometry}

\usepackage[T1]{fontenc}

\usepackage[colorlinks=true,linkcolor=Turquoise4,citecolor=BlueViolet]{hyperref}
\usepackage{url}
\usepackage{mathtools}
\usepackage{amssymb}
\usepackage[nameinlink, capitalise]{cleveref}
\usepackage{graphicx}
\usepackage{adjustbox}
\usepackage{multirow}
\usepackage{natbib}
\usepackage[dvipsnames,x11names]{xcolor}
\usepackage{subcaption}
\usepackage{lmodern}
\usepackage{pgf}
\usepackage{pgfplots}

\newcommand{\ket}[1]{\ensuremath{|#1\rangle}}
\newcommand{\bra}[1]{\ensuremath{\langle #1|}}
\newcommand{\braket}[2]{\ensuremath{\langle #1|#2\rangle}}
\newcommand{\ketbra}[2]{\ensuremath{|#1\rangle\!\langle #2|}}

\newcommand{\norm}[1]{\left\lVert #1 \right\rVert}

\DeclareMathOperator{\tr}{\text{tr}}
\DeclareMathOperator{\var}{Var}
\DeclareMathOperator{\cov}{Cov}
\DeclareMathOperator{\dir}{Dir}

\newenvironment{theorem*}[1]%
  {\par\noindent\textbf{Theorem #1. }\itshape}%
  {\par\vspace{1ex}}

\newenvironment{lemma*}[1]%
  {\par\noindent\textbf{Lemma #1. }\itshape}%
  {\par\vspace{1ex}}

\newenvironment{corollary*}[1]%
  {\par\noindent\textbf{Corollary #1. }\itshape}%
  {\par\vspace{1ex}}

\begin{document}

\title{Trainability of Parametrised Linear Combinations of Unitaries}
\author{Nikhil Khatri\inst{1,2}, Stefan Zohren\inst{1}, Gabriel Matos\inst{2}}
\institute{Machine Learning Research Group, Department of Engineering Science, University of Oxford \\
Eagle House, Walton Well Road, Oxford, United Kingdom
\and
Quantinuum, Partnership House, Carlisle Place, London SW1P 1BX, United Kingdom 
}

\maketitle
\begin{abstract}
A principal concern in the optimisation of parametrised quantum circuits is the presence of barren plateaus, which present fundamental challenges to the scalability of applications, such as variational algorithms and quantum machine learning models. Recent proposals for these methods have increasingly used the linear combination of unitaries (LCU) procedure as a core component. In this work, we prove that an LCU of trainable parametrised circuits is still trainable. We do so by analytically deriving the expression for the variance of the expectation when applying the LCU to a set of parametrised circuits, taking into account the postselection probability. These results extend to incoherent superpositions.
We support our conclusions with numerical results on linear combinations of fermionic Gaussian unitaries (matchgate circuits). Our work shows that sums of trainable parametrised circuits are still trainable, and thus provides a method to construct new families of more expressive trainable circuits. 
We argue that there is a scope for a quantum speed-up when evaluating these trainable circuits on a quantum device.

\end{abstract}

\section{Introduction}
\label{sec:introduction}
Most applications relying on the optimisation of parametrised quantum circuits, such as quantum machine learning and variational algorithms, are known to suffer from the barren plateau phenomenon~\citep{mcclean2018barren, Larocca2025}: an exponential decrease in the variance of the expectation value with an increase in the number of qubits. This poses a fundamental challenge to the scalability of these methods. Thus, it is of great value to find parametrised circuits that are \emph{trainable} at scale, i.e., have a polynomially decreasing variance of the expectation with the number of qubits.

In this work, we study the linear combinations of unitaries (LCU) construction in a variational context. Originally proposed for Hamiltonian simulation \citep{childs2012}, these circuits have recently found use in quantum machine learning, where the unitaries or the coefficients in the superposition are parametrised and updated as part of an optimisation routine \citep{coopmans2024sample,khatri2024quixer,heredge_nonunitary_qml}. Despite these recent applications, there has been no rigorous characterisation of the variance of expectation values in variational LCUs, leaving open the question of the conditions under which these are trainable. Here, we fill this gap by deriving an analytical expression for the variance of the expectation value of these circuits. To this end, we model the distribution of LCU coefficients using the uniform Dirichlet distribution, the maximum entropy distribution over the space of $l_1$-normalised vectors. We also derive a general bound independent of the distribution of the LCU coefficients. These results take into account the postselction probability of the LCU.

Our work shows that linear combinations of trainable parametrised circuits are still trainable, and precisely quantifies the degree to which each component of the LCU contributes to the variance of the expectation values. To demonstrate this, we specialise our analytical results to the case of fermionic Gaussian unitaries (also known as matchgate circuits), and numerically compute the variance of their LCU, observing a close match with our analytical expression. This case is particularly relevant, as matchgates are known to be trainable~\citep{diaz2023}, while linear combinations of matchgates become more and more expressive as the number of terms increases. This yields a family of trainable circuits with increasing expressivity, which is controlled by the number of terms of the LCU. We remark that the gate cost of running these matchgate LCU circuits on a quantum computer is polynomially lower than the number of operations necessary to simulate them classically, and thus there exists the scope for a quantum speed-up in trainable parametrised quantum circuits. We note, however, that for current hardware, the time it takes to execute a classical operation is much lower than the time it takes to run a quantum gate, making this improvement currently impractical.  \\

\noindent \emph{Outline:} In \cref{sec:analytical}, we present our main results, consisting of an  analytical expression for the variance of an expectation value of an observable with respect to the output state of an LCU and a more general bound on this variance. We then consider the special case where the unitaries in the linear combination are fermionic Gaussian unitaries in \cref{sec:ff_lcu}, and compare this with numerical results in \cref{sec:numerical}. Finally, we extend our results to incoherent superpositions in \cref{sec:incoherent_superpositions}, before concluding in \cref{sec:conclusion}.

\section{Main Results}
\label{sec:analytical}
In this section, we summarise our main results, which consist of:
\begin{itemize}
    \item[$\circ$] A lower bound on the variance of the expectation, independent of the distribution of the LCU coefficients, showing that an LCU of trainable parametrised circuits is still trainable (\cref{thm:general_bound}); 
    \item[$\circ$] The exact expression for this variance when the LCU coefficients follow a uniform distribution over the space of $l_1$-normalised vectors (\cref{thm:var_dirichhaar}), giving the individual contributions of each LCU term to the variance;
    \item[$\circ$] Versions of these results that apply to incoherent superpositions of mixed states (see \cref{sec:incoherent_superpositions}).
\end{itemize}
We now lay out the setup necessary to formally state these results, with proofs given in \cref{sec:proofs}. 

Given a set of parametrised quantum circuits $U_1(\vec{\theta^{(1)}}), ..., U_k(\vec{\theta^{(k)}})$ on $N$ qubits and parameters $\vec{c}$ such that $\norm{c}_1 = 1$, the LCU construction~\citep{childs2012} yields the state 
\begin{align}
\label{eq:lcu_output}
\ket{\psi} = \sum_{j=1}^k c_j\cdot U_j(\vec{\theta^{(j)}}) \ket{0}.  
\end{align}
Note that the $l_1$-normalisation is a requirement of the LCU. In what follows, we will assume each unitary to be sampled from the Haar distribution over the \emph{Lie group} it generates; see \cref{sec:lie_groups} for a brief review of these structures. This is justified by sufficiently deep parametrised circuits approximating the Haar distribution over these groups~\citep{DAlessandro2021}, and follows previous approaches~\citep{Ragone2024, Fontana2024}. These Haar distributions are independent, since the parameters $\vec{\theta^{(j)}}$ are also taken to be independent.

Thus, we consider states of the form 
\begin{align}
    \ket{\psi} = \sum_{j=1}^k c_j\cdot U_j \ket{0}
\end{align}
where $U_j$ is sampled from the Haar distribution over the corresponding Lie group $\mathcal{G}_j$, and the coefficients are sampled from a distribution on the space of $l_1$-normalised vectors. Note that a linear combination of unitary matrices is not in general unitary, and the state $\ket{\psi}$ may be sub-normalized. This sub-normalisation corresponds to the postselection probability in the LCU, which is taken into account in our approach. This is because we do not normalise the resulting state \eqref{eq:lcu_output} in our derivations, so that this sub-normalisation is directly reflected in the expectation \eqref{eq:cost}.

The quantities of interest in applications typically take the form of an expectation value $m$ of some Hermitian observable $O$,
\begin{align}
\label{eq:cost}
    m := \tr(\rho O),
\end{align}
where $\rho = \ketbra{\psi}{\psi}$ for pure states. By linearity, $m$ decomposes as a sum 
\begin{align}
    m &= \sum_{i,j=1}^k c_i c_j m_{ij}\\
    m_{ij}&:=~\tr( U_i\rho_0 U_j^{\dagger} O)
\end{align}
including $k$ terms of the form 
\begin{align}
m_i := m_{ii}.
\end{align}
We begin by deriving a bound on the variance of \eqref{eq:cost} making no assumptions on the distribution of the LCU coefficients.
\begin{theorem}
\label{thm:general_bound}
The variance of the expectation value of an observable $O$ for an LCU of Haar random unitaries, for which $E[m_i] = 0$, is lower bounded by
\begin{align}
    \var[m] &\geq \sum_{j=1}^k E[c_j^2]^2 \var[m_j] \geq \frac{1}{k^3}\min_j \var[m_j].
\end{align}
\end{theorem}
\begin{proof}
    \hyperref[proof:general_bound]{Proof in \cref*{sec:lcu_haar}.}
\end{proof}
The constraint in \cref{thm:general_bound} encompasses most relevant observables. For $\mathcal{G}_j = SU(2^N)$, any traceless observable (equivalently, any linear combination of Pauli strings not including the identity) satisfies $E[m_j] = 0$ \citep{mele2024introduction}. For $\mathcal{G}_j = SO(2N)$, any linear combination of Pauli strings not supported on all qubits and not including the identity satisfies $E[m_i] = 0$~\citep{diaz2023}. This is a general bound, with no restrictions on the distribution on the $l_1$-normalised LCU weights, and represents a \emph{worst-case} setting.

In what follows, we assume $\vec{c}$ to follow a uniform Dirichlet distribution. This is the maximum entropy distribution over the space of $l_1$-normalised vectors, and corresponds to the measure being distributed as uniformly as possible over the underlying space; this is analogous to the Haar measure being the maximum entropy distribution over the corresponding Lie group.

\begin{theorem}\label{thm:var_dirichhaar}
For state a $\rho$ prepared by a uniform Dirichlet linear combination of Haar random unitaries $\sum_{j=1}^k c_j\cdot U_j \ket{0}$, the variance of the expectation value $\tr(\rho O)$ of an observable $O$ is given by
\begin{align}
\label{eq:dirichlet_lcu_var_general}
\var[m] = \sum_{i=1}^k \frac{24 \cdot E[m_i^2]}{k \cdot (k+1) \cdot (k+2) \cdot (k+3)} + \sum_{i\neq j}^k \frac{4 \cdot (E[m_i]E[m_j] + E[|m_{ij}|^2])}{k \cdot (k+1) \cdot (k+2) \cdot (k+3)} - \sum_{i,j=1}^k \frac{4 \cdot E[m_i]E[m_j]}{k^2 (k+1)^2}.
\end{align}
Specialising to the case of where all Lie groups are the same yields
\begin{align}
\label{eq:dirichlet_lcu_var}
\var[m] = \frac{24 \cdot E[m_i^2]}{(k+1) \cdot (k+2) \cdot (k+3)} + \frac{4 \cdot (k-1) \cdot (E[m_i]^2 + E[|m_{ij}|^2])}{(k+1) \cdot (k+2) \cdot (k+3)} - \frac{4 \cdot E[m_i]^2}{(k+1)^2}.
\end{align}
\end{theorem}
\begin{proof}
    \hyperref[proof:var_dirichhaar]{Proof in \cref*{sec:lcu_dirichlet_haar}.}
\end{proof}
The terms $E[m_{i}^2], E[m_{i}]^2, E[|m_{ij}|^2]$ can be computed using the Weingarten calculus and the Haar measure on the Lie groups $\mathcal{G}_i, \mathcal{G}_j$, yielding analytical expressions for these, and therefore also for the variance of the expectation value. 

\section{Examples}

\subsection{Linear combinations of fermionic Gaussian unitaries / matchgate circuits}
\label{sec:ff_lcu}

 A particularly relevant case is that of linear combinations of fermionic Gaussian unitaries (also called matchgate circuits). This is because matchgate circuits are known to be trainable~\citep{diaz2023}, while linear combinations of matchgates become more and more expressive as the number of terms increases. Thus, applying our results to matchgates yields a family of trainable circuits with expressivity controlled by the number of terms of the LCU. 

Fermionic Gaussian unitaries correspond to the Lie group $\mathcal{G} = SO(2N)$. The number of terms $k$ in the LCU upper bounds the \emph{fermionic Gaussian rank} of the resulting state, which is the minimum number of fermionic Gaussian states that a state can be expressed as the linear combination of. We begin by deriving the value of $E[|m_{ij}|^2]$ for this Lie group.

\begin{lemma}
\label{lem:ff_abssq}
If the Lie group is the space of all fermionic Gaussian unitaries on $N$ qubits i.e.\ $\mathcal{G} = SO(2N)$, and $O$ is an observable, then
\begin{align}
    E[|m_{ij}|^2] &= \frac{1}{2^{2N}} \tr(\rho_0^2)\tr(O^2) + \frac{2}{2^{2N}} \tr(\rho_0^2P)\tr(O^2P) + \frac{1}{2^{2N}} \tr(\rho_0P\rho_0P)\tr(OPOP)
\end{align}
where $P$ is the fermionic parity operator.
\end{lemma}
\begin{proof}
    \hyperref[proof:ff_abssq]{Proof in \cref*{sec:app_lcu_ff}.}
\end{proof}
We leverage the work of \citet{diaz2023} in \cref{sec:rank1exps} to obtain $E[m_{i}^2], E[m_{i}]^2$. In the following result, we consider observables that are quadratic Hamiltonians, which describe free fermion systems (see \cref{sec:free_fermions}). We give a more general version of \cref{cor:ffvar} lifting the constraints on the observable and the initial state as \cref{cor:general_ffvar_app} of \cref{sec:app_lcu_ff}.
\begin{corollary}\label{cor:ffvar}
The expectation value of a quadratic observable $O$ for an LCU of Haar random fermionic Gaussian unitaries with an even, pure fermionic Gaussian initial state, has variance 
\begin{align}
\label{eq:ffvar}
    \var[m] = \frac{\tr(O^2)}{2^N} \left [ \frac{24}{(k+1) \cdot (k+2) \cdot (k+3) \cdot (2N-1)} + \frac{4 \cdot (k-1)}{(k+1) \cdot (k+2) \cdot (k+3) \cdot 2^{N-1} } \right ].
\end{align}
\end{corollary}
\begin{proof}
    \hyperref[proof:ffvar]{Proof in \cref*{sec:app_lcu_ff}.}
\end{proof}
Note that $\tr(O^2)/2^N = 1$ for most relevant observables. For instance, if $S$ is a Pauli string, $\tr(S^2)/2^N = \tr(I)/2^N = 2^N/2^N = 1$, where we have used the fact that Pauli strings are self-inverse. 
\begin{corollary}\label{cor:bigoZ1}
The variance of the expectation value of an observable for an LCU of Haar random fermionic Gaussian unitaries scales as
\begin{align}
    \var[m] \in \Omega \left (\frac{1}{N^s \cdot k^3} \right ).
\end{align}
where $m$ depends on the observable. For quadratic observables, $s=1$.
\end{corollary}
\begin{proof}
    \hyperref[proof:bigoZ1]{Proof in \cref*{sec:app_lcu_ff}}; the case for quadratic observables follows directly from \cref{cor:ffvar}.
\end{proof}
We conclude that, for linear combinations of fermionic Gaussian states, the gradient has a lower bound polynomial in both the number of qubits, and the rank of the system. We comment on the significance of this in \cref{sec:conclusion}.

\subsection{Linear combinations of expressive unitaries}

In this section, we state results for the case where all Lie groups are equal to $\mathcal{G} = SU(2^N)$ i.e. the case where every unitary is expressive.
\begin{lemma}
\label{lem:abssquare}
If the Lie group is the space of all unitary matrices on $N$ qubits i.e. $\mathcal{G} = SU(2^N)$, and $O$ is an observable, then
\begin{align}
    E[|m_{ij}|^2] = \frac{\tr(O^2)}{2^{2N}}
\end{align}  
\end{lemma}
\begin{proof}
\hyperref[proof:abssquare]{Proof in \cref*{sec:app_lcu_expressive}.}
\end{proof}

\begin{corollary}\label{cor:expressive_var}
The expectation value of an observable $O$ for an LCU of Haar random unitaries has variance 
\begin{align}
\label{eq:expressive_var}
   \var[m] &= \frac{24 \cdot (\tr(O^2) - tr(O)^2/2^N) \cdot (\tr(\rho_0^2) - 1/2^N)}{(4^N -1) \cdot (k+1) \cdot (k+2) \cdot (k+3)} \\ &+ \frac{4 \cdot (k-1) \cdot (\tr(O^2)/ 2^N\tr(\rho_0^2) + \tr(O)^2)}{(k+1) \cdot (k+2) \cdot (k+3) \cdot 2^N} \\ &- \frac{\tr(O)^2}{2^N}.
\end{align}
\end{corollary}
\begin{proof}
\hyperref[proof:expressive_var]{Proof in \cref*{sec:app_lcu_expressive}.}
\end{proof}
\section{Numerical Results}
\label{sec:numerical}

There are polynomial-time classical algorithms for simulating free fermion systems, up to a phase, based on the covariance matrix formalism~\citep{terhal2002}. Recent work has extended this framework to include a phase, extending classical simulation techniques to coherent superpositions \citep{dias2024classical,reardon2024improved, cudby2024gaussiandecompositionmagicstates}. In our simulations, we use the method of \citet{dias2024classical}, where each state $\ket{\psi}$ is associated with a basis state \ket{x}, and the complex overlap $\braket{x}{\psi}$, which is updated throughout the evolution of the circuit. We implement their algorithm in Python using JAX \citep{jax2018github}, and use this to compute our experimental results.

\begin{figure}[ht]
    \centering
    \includegraphics[width=0.6\linewidth]{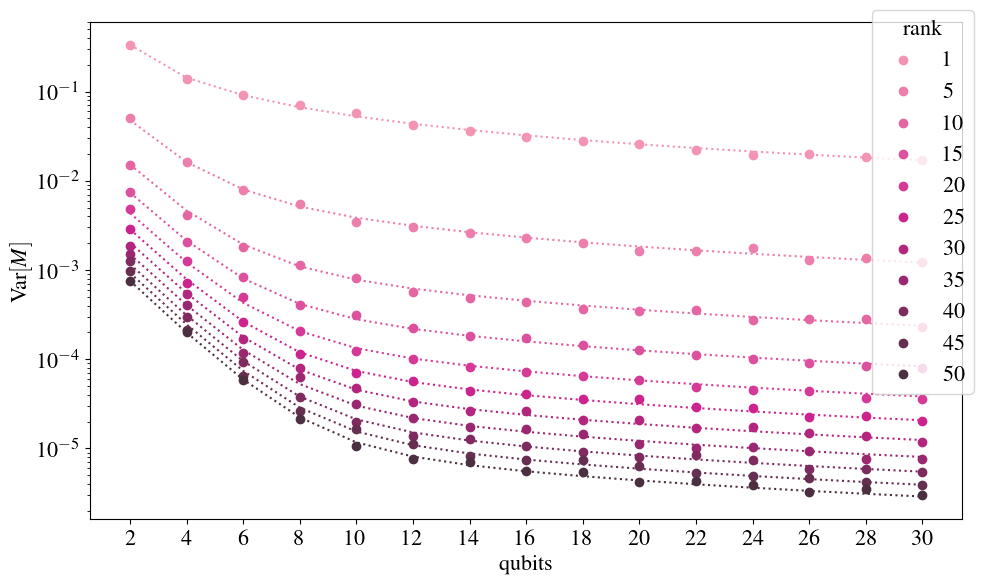}
    \caption{Variance of the $Z_1$ expectation value for linear combinations of Haar random fermionic Gaussian states. Dotted lines indicate the analytical value from \cref{eq:ffvar}, and points are computed numerically from 1000 samples. The label refers to the number of terms $k$ in the LCU. Analytical and numerical results show close agreement.}
    \label{fig:var_by_qb}
\end{figure}

We prepare coherent superpositions of fermionic Gaussian states, where each state is sampled from the distribution induced on states by the $SO(2N)$ Haar measure, and the coefficients are sampled from the uniform Dirichlet distribution. 

We numerically compute the expectation of $O = Z_1$, which from \eqref{eq:ffvar} has variance
\begin{align}
\label{eq:z1var}
    \var[m] = \frac{24}{(k+1) \cdot (k+2) \cdot (k+3) \cdot (2N-1)} + \frac{4 \cdot (k-1)}{(k+1) \cdot (k+2) \cdot (k+3) \cdot (2^{N-1})}.
\end{align}
The variance of this expectation value is calculated over $1000$ samples. This experiment is repeated for varying system sizes and fermionic Gaussian ranks. \Cref{fig:var_by_qb} plots the variances obtained numerically, and compares these with values from the analytical expression in \cref{eq:ffvar}; we observe close agreement between these.

\section{Incoherent superpositions}
\label{sec:incoherent_superpositions}
Analogous results also hold in the simpler case when we take incoherent superpositions of quantum states
\begin{align}
    \label{eq:app_incoherent_superposition}
    \widetilde{m} = \sum_i c_i \tr(U_i \rho_i U_i^\dagger O) = \sum_i c_i m_i.
\end{align}

\begin{theorem}
\label{thm:general_incoherent_bound}
The variance of the expectation value of an observable $O$ for an incoherent superposition of Haar random unitaries, for which $E[m_i] = 0$, is lower bounded by
\begin{align}
    \var(\widetilde{m}) &\geq \sum_{i=1}^k E[c_i]^2 \var[m_i] \geq \frac{1}{k}\min_i \var[m_i].
\end{align}
\end{theorem}
\begin{proof}
\hyperref[proof:general_incoherent_bound]{Proof in \cref*{sec:incoherent_general_bound}.}
\end{proof}

\begin{theorem}
\label{thm:incoherent_var_dirichhaar}
For a state $\rho$ given by incoherent superposition of Haar random unitaries applied to some initial states where the coefficients are distributed according to a uniform Dirichlet distribution, the variance of the expectation value $\tr(\rho O)$ of an observable $O$ is given by
\begin{align}
    \var(\widetilde{m}) &= \sum_{i=1}^k \frac{2}{k(k+1)}  E[m_i^2] - \frac{1}{k^2}  E[m_i]^2 - \sum_{i\neq j}^k \frac{1}{k^2(k+1)}\  E[m_i] E[m_j],
\end{align}
Specialising to the case of where all Lie groups are the same yields
\begin{align}
    \var(\widetilde{m}) &= \frac{2}{(k+1)}  \var(m_i)
\end{align}
\end{theorem}
\begin{proof}
\hyperref[proof:incoherent_var_dirichhaar]{Proof in \cref*{sec:incoherent_general_bound}.}
\end{proof}

\section{Conclusion}
\label{sec:conclusion}

Our primary contributions are an analytical derivation of the variance of the expectation value for states obtained from a linear combinations of unitaries, including the particular case of the variance for systems with upper bounded fermionic Gaussian rank. Thus, given any trainable circuit, we are able to construct a trainable family of circuits controlled by a rank parameter $k$.

The family of circuits we consider can be simulated classically, as is the case for many trainable parametrised circuits in the literature~\citep{cerezo_classical_2024}. An open question is to prove that this is the case for \emph{any} linear combination of trainable circuits. Importantly, the cost of classical simulation for fermionic Gaussian states with rank $k$ is $\mathcal{O}\left( k^2 \cdot N^3\right)$~\citep{dias2024classical}, while the gate count of an efficient quantum implementation is $\mathcal{O}(k\cdot N^2)$~\citep{unary_iteration, kokcu_compression, jiang_fgs_preparation}; thus, there is still potential for a speed-up when running these circuits on a quantum computer. It is important to note, however, that for current hardware, the time it takes to execute a classical operation is much lower than the time it takes to run a quantum gate. Thus, despite there being an asymptotic improvement, it is at present not a practical one.

\section{Acknowledgements}

We thank Eric Brunner and Enrico Rinaldi for reviewing this manuscript and for their insightful comments.

\bibliographystyle{plainnat}
\bibliography{references}

\appendix
\section{Derivations of main results}
\label{sec:proofs}
Assume $U_1,...,U_k$ to be unitaries sampled independently according to distributions $\mu_1,...,\mu_k$. Let $\vec{c} = [c_1,...,c_k]$ be an $l_1$-normalised positive vector sampled according to some distribution independently from the aforementioned unitaries. Let $O$ be an observable and $\rho_0$ be some initial quantum state. We are interested in characterising the random variable
\begin{align}
    m &:= \tr \left (U\rho_0 U^\dagger O \right ), \\
    U &= \sum_{j=1}^k c_j U_j,
\end{align}
which corresponds to the expectation value for an observable $O$ for a state prepared by the Linear Combination of Unitaries (LCU) procedure \citep{childs2012}.
For ease of notation, we define
\begin{align}
    m_{ij} :=& \tr(U_i \rho_0 U_j^\dagger O), \\
    m_i :=& m_{ii},
\end{align}
and express the random variable of interest as
\begin{align}
    m = \sum_{i,j}c_i c_j m_{ij}. \label{eq:simplified_trace}
\end{align}
To determine the trainability of such LCU systems, we are interested in characterising the variance of this random variable, given by
\begin{align}
    \var[m] = E[m^2] - E[m]^2.
\end{align}
Note that, since we are assuming the unitaries to be independent from $\vec{c}$, the expectation value takes the form
\begin{align}
\label{eq:general_expectation}
    E[m] =& \sum_{i,j}E[c_i c_j] \cdot E[m_{i,j}]
\end{align}
and the variance consists of two terms,
\begin{align}
    E[m]^2 =& \sum_{i,j,k,l}E[c_i c_j]E[c_k c_l] \cdot E[m_{ij}] E[m_{kl}]\\
    E[m^2] =& \sum_{i,j,k,l}E[c_i c_j c_k c_l] \cdot E[m_{ij}m_{kl}].
\end{align}
In the following sections we provide closed form expressions for these terms in specific parameter settings.

\subsection{LCUs of Haar distributed unitaries}
We first make the assumption common in the parametrised circuit trainability literature, that each unitary is \textit{i.i.d.} sampled from the Haar measure over some Lie group $\mathcal{G}$, i.e. $U_i \sim Haar(\mathcal{G})$. Using the invariance of the Haar measure under left multiplication by a fixed unitary, we obtain
\begin{lemma}\label{lem:emij_zero}
    \begin{align}
    E[m_{ij}^n] = 0, \qquad n \geq 1, i \neq j
\end{align}
\end{lemma}
\begin{proof}
\begin{align}
    E[m_{ij}^n] =& E[\tr(U_i \rho_0 U_j^\dagger O)^n]  \\
 =& E[\tr((e^{\frac{i\pi}{n}}I)U_i \rho_0 U_j^\dagger O)^n] \\ =& -E[m_{ij}^n]\\
    \implies &E[m_{ij}^n] = 0
\end{align}
\qed
\end{proof}
Using \cref{lem:emij_zero}, the expectation value \eqref{eq:general_expectation} of $m$ may be simplified by removing cross terms, giving
\begin{align}
\label{eq:general_e2_m}
    E[m]^2 =& \sum_{i,j=1}^k E[c_i^2]E[c_j^2] \cdot E[m_i] E[m_j].
\end{align}
Likewise, $E[m^2]$ may be reduced to the following non-zero terms
\begin{align}
\label{eq:general_e_m2}
    E[m^2] =& \sum_{i=1}^k E[c_i^4] \cdot E[m_{i}^2] \nonumber \\
           +& \sum_{i \neq j}^k E[c_i^2 c_j^2] \cdot E[m_i m_j] \nonumber\\
           +& \sum_{i \neq j}^k E[c_i^2 c_j^2] \cdot E[|m_{i,j}|^2].
\end{align}

\subsection{Bound on variance for LCUs of Haar distributed unitaries}
\label{sec:lcu_haar}

\begin{lemma}
\label{lem:holder}
For any distribution on the space of $l_1$-normalised vectors, it is the case that
\begin{align}
    \sum_{i=1}^k E[c_i]^n \geq \frac{1}{k^{n-1}} 
\end{align}
\end{lemma}
\begin{proof}
By Holder's inequality,
\begin{align}
    &\sum_{i=1}^k \left ( (E[c_i]^n)^\frac{1}{n} \cdot 1^{1-\frac{1}{n}} \right )^1 \leq \left ( \sum_{i=1}^k (E[c_i]^n)^1 \right )^{1/n} \cdot \left ( \sum_{i=1}^k 1^1 \right )^{1-\frac{1}{n}} \\
    &\iff \sum_{i=1}^k E[c_i] \leq \left ( \sum_{i=1}^k E[c_i]^n \right )^{1/n} \cdot k^{\frac{n-1}{n}} \\
    &\iff 1 \leq \left ( \sum_{i=1}^k E[c_i]^n \right ) \cdot k^{n-1} \\ 
    &\iff \frac{1}{k^{n-1}} \leq \sum_{i=1}^k E[c_i]^n  
\end{align}
\qed
\end{proof}

\begin{theorem*}{\ref{thm:general_bound}}
The variance of the expectation value of an observable $O$ for an LCU of Haar random unitaries, for which $E[m_i] = 0$, is lower bounded by
\begin{align}
    \var[m] &\geq \sum_{j=1}^k E[c_j^2]^2 \var[m_j] \geq \frac{1}{k^3}\min_j \var[m_j].
\end{align}
\end{theorem*}
\begin{proof}
\label{proof:general_bound}
From the previous section, the variance may be stated as
\begin{align}
    \var[m] =& \sum_{i=1}^k E[c_i^4] E[m_i^2] + \sum_{i \neq j}^k E[c_i^2 c_j^2] E[m_i] E[m_j] + \sum_{i \neq j}^k E[c_i^2 c_j^2] \cdot E[|m_{i,j}|^2] \nonumber\\ &- \sum_{i,j=1}^k E[c_i^2] E[c_j^2] E[m_i] E[m_j] \\
    &\geq \sum_{i=1}^k E[c_i^4] E[m_i^2] - E[c_i^2]^2 E[m_i]^2 + \sum_{i \neq j}^k (E[c_i^2 c_j^2] - E[c_i^2] E[c_j^2]) E[m_i] E[m_j] \\
    &= \sum_{i=1}^k E[c_i^4] E[m_i^2] - E[c_i^2]^2 E[m_i]^2 + \sum_{i \neq j}^k \cov(c_i^2, c_j^2) E[m_i] E[m_j] \label{eq:variance} \\
    &\geq \sum_i^k E[c_i]^4 \var[m_i] + \sum_{i \neq j}^k \cov(c_i^2, c_j^2) E[m_i] E[m_j] \label{eq:jensen} \\
    &\geq \frac{1}{k^3} \min_i \var[m_i] + \sum_{i \neq j}^k \cov(c_i^2, c_j^2) E[m_i] E[m_j] \label{eq:holder}
\end{align}
Where in \eqref{eq:jensen} we have used Jensen's inequality and \eqref{eq:holder} we have used \cref{lem:holder}. If $E[m_j]=0$, which happens e.g. if the observable is a Pauli matrix in the case of the unitary group, or if it is a Pauli not spanning all qubits in the case of the special orthogonal group under the Jordan-Wigner transform, then
\begin{align}
    \var[m] &\geq \sum_{j=1}^k E[c_j^2]^2 \var[m_j] \geq \frac{1}{k^3} \min_i \var[m_i].
\end{align}
\qed
\end{proof}

\subsection{Haar LCUs with uniform Dirichlet coefficients}
\label{sec:lcu_dirichlet_haar}

Adding the condition that the coefficients of the linear combination are sampled from a uniform Dirichlet distribution, $\vec{c} \sim \mathrm{Dir}(1 \dots 1)$, we obtain the statement of \cref{thm:var_dirichhaar}, which we restate and prove below.
\begin{theorem*}{\ref{thm:var_dirichhaar}}
\label{thm:var_dirichhaar_app}
For state a $\rho$ prepared by a uniform Dirichlet linear combination of Haar random unitaries $\sum_{j=1}^k c_j\cdot U_j \ket{0}$, the variance of the expectation value $\tr(\rho O)$ of an observable $O$ is given by
\begin{align}
\label{eq:dirichlet_lcu_var_general_app}
\var[m] = \sum_{i=1}^k \frac{24 \cdot E[m_i^2]}{k \cdot (k+1) \cdot (k+2) \cdot (k+3)} + \sum_{i\neq j}^k \frac{4 \cdot (E[m_i]E[m_j] + E[|m_{ij}|^2])}{k \cdot (k+1) \cdot (k+2) \cdot (k+3)} - \sum_{i,j=1}^k \frac{4 \cdot E[m_i]E[m_j]}{k^2 (k+1)^2}.
\end{align}
Specialising to the case of where all Lie groups are the same yields
\begin{align}
\label{eq:dirichlet_lcu_var_app}
\var[m] = \frac{24 \cdot E[m_i^2]}{(k+1) \cdot (k+2) \cdot (k+3)} + \frac{4 \cdot (k-1) \cdot (E[m_i]^2 + E[|m_{ij}|^2])}{(k+1) \cdot (k+2) \cdot (k+3)} - \frac{4 \cdot E[m_i]^2}{(k+1)^2}.
\end{align}
\end{theorem*}
\begin{proof}
\label{proof:var_dirichhaar}
Taking the coefficients to be drawn from the uniform Dirichlet distribution, i.e. $\vec{c} \sim \mathrm{Dir}(1 \dots 1)$, closed form expressions for the moments of the coefficients (\cref{sec:dirich}) help reduce \eqref{eq:general_e2_m}, \eqref{eq:general_e_m2} to
\begin{align}
    E[m]^2 =& \sum_{i,j=1}^k \frac{4}{k^2(k+1)^2} \cdot E[m_i]E[m_j]
\end{align}
\begin{align}
    E[m^2] =& \sum_{i=1}^k \frac{24}{k (k+1) (k+2) (k+3)} \cdot E[m_{i}^2]\nonumber\\
           +& \sum_{i \neq j}^k \frac{4}{k (k+1) (k+2) (k+3)} \cdot E[m_i m_j]\nonumber\\
           +& \sum_{i \neq j}^k \frac{4}{k (k+1) (k+2) (k+3)} \cdot E[|m_{i,j}|^2].
\end{align}
This yields an expression for the variance as
\begin{align}
    \var[m] = \sum_{i=1}^k \frac{24 \cdot E[m_{i}^2]}{k (k+1) (k+2) (k+3)} 
           + \sum_{i \neq j}^k \frac{4 \cdot (E[m_i]E[m_j] + E[|m_{i,j}|^2])}{k (k+1) (k+2) (k+3)} - \sum_{i,j=1} \frac{4 \cdot E[m_i]E[m_j]}{k^2 (k+1)^2}.
\end{align}
Specialising to the case where all Lie groups are the same yields
\begin{align}
    \var[m] = \frac{24 \cdot E[m_{i}^2]}{(k+1) (k+2) (k+3)} 
           + \frac{4 \cdot (k-1)  \cdot (E[m_i]^2 + E[|m_{i,j}|^2])}{(k+1) (k+2) (k+3)} - \frac{4 \cdot E[m_i]^2}{(k+1)^2}.
\end{align}
\qed
\end{proof}

\subsection{Bound on variance for incoherent superpositions of Haar distributed unitaries}
\label{sec:incoherent_general_bound}
Consider the expectation value of an incoherent superposition of states prepared by parametrised quantum circuits
\begin{align}
    \widetilde{m} = \sum_i c_i \tr(U_i \rho_i U_i^\dagger O) = \sum_i c_i m_i.
\end{align}
Then,
\begin{align}
    E[\widetilde{m}] &= \sum_i E[c_i] E[m_i], \label{eq:incoherent_expectation} \\
    \var(\widetilde{m}) &= \sum_{i=1}^k E[c_i^2]  E[m_i^2] - E[c_i]^2  E[m_i]^2 + \sum_{i\neq j}^k \cov(c_i, c_j)  E[m_i] E[m_j], \label{eq:incoherent_variance}
\end{align}
and we have the following result.
\begin{theorem*}{\ref{thm:general_incoherent_bound}}
The variance of the expectation value of an observable $O$ for an incoherent superposition of Haar random unitaries, for which $E[m_i] = 0$, is lower bounded by
\begin{align}
    \var(\widetilde{m}) &\geq \sum_{j=1}^k E[c_j]^2 \var[m_j] \geq \frac{1}{k}\min_j \var[m_j].
\end{align}
\end{theorem*}
\begin{proof}
\label{proof:general_incoherent_bound}
We have that
\begin{align}
    \var(\widetilde{m}) &= \sum_{i=1}^k E[c_i^2]  E[m_i^2] - E[c_i]^2  E[m_i]^2 + \sum_{i\neq j}^k \cov(c_i, c_j)  E[m_i] E[m_j] \\
    &\geq \sum_{i=1}^k E[c_i]^2  \var(m_i) + \sum_{i\neq j}^k \cov(c_i, c_j)  E[m_i] E[m_j] \label{eq:general_incoherent_bound_1} \\
    &\geq \frac{1}{k} \min_i \var(m_i) + \sum_{i\neq j}^k \cov(c_i, c_j)  E[m_i] E[m_j], \label{eq:general_incoherent_bound_2}
\end{align}
where in \eqref{eq:general_incoherent_bound_1} we used Jensen's inequality and in \eqref{eq:general_incoherent_bound_2} we used \cref{lem:holder}.
\qed
\end{proof}

\subsection{Haar incoherent superpositions with uniform Dirichlet coefficients}

\begin{theorem*}{\ref{thm:incoherent_var_dirichhaar}}
For a state $\rho$ given by incoherent superposition of Haar random unitaries applied to some initial states where the coefficients are distributed according to a uniform Dirichlet distribution, the variance of the expectation value $\tr(\rho O)$ of an observable $O$ is given by
\begin{align}
\label{eq:app_general_incoherent_var_dirichhaar}
    \var(\widetilde{m}) &= \sum_{i=1}^k \frac{2}{k(k+1)}  E[m_i^2] - \frac{1}{k^2}  E[m_i]^2 - \sum_{i\neq j}^k \frac{1}{k^2(k+1)}\  E[m_i] E[m_j],
\end{align}
Specialising to the case of where all Lie groups are the same yields
\begin{align}
\label{eq:app_incoherent_var_dirichhaar}
    \var(\widetilde{m}) &= \frac{2}{(k+1)}  \var(m_i)
\end{align}
\end{theorem*}
\begin{proof}
\label{proof:incoherent_var_dirichhaar}
\eqref{eq:app_general_incoherent_var_dirichhaar} follows directly from replacing \eqref{eq:dirichlet_e}, \eqref{eq:dirichlet_e2} from \cref{sec:dirichlet_moments} into \eqref{eq:incoherent_expectation} from \cref{sec:incoherent_general_bound}. Specialising to the case where all Lie groups are equal, we obtain
\begin{align}
    \var(\widetilde{m}) &= \frac{2}{k+1}  E[m_i^2] - \frac{1}{k}  E[m_i]^2 - \frac{k-1}{k(k+1)}\  E[m_i]^2, \\
    &= \frac{2}{k+1}  E[m_i^2] - \frac{k + 1 + k - 1}{k(k+1)}\  E[m_i]^2, \\
    &= \frac{2}{k+1}  E[m_i^2] - \frac{2k}{k(k+1)}\  E[m_i]^2, \\
    &= \frac{2}{k+1}  E[m_i^2] - \frac{2}{k+1}\  E[m_i]^2, \\
    &= \frac{2}{k+1}  \var (m_i).
\end{align}
\qed
\end{proof}

\section{Dirichlet Moments}\label{sec:dirich}
\label{sec:dirichlet_moments}
Here, we consider moments of a random variable $c$ distributed on the space of positive $l_1$ normalised vectors according to the Dirichlet distribution $c = (c_1,...c_k) \sim \dir(\vec{\alpha})$, which we assume to be uniform, i.e. $\alpha = (1,...,1)$. Define $\alpha_0 := \sum \alpha_i$ and $\beta_i := \alpha_0 - \alpha_i$. For the Dirichlet distribution, the marginal distribution of the $i$th coordinate $c_i$ is a Beta$(\alpha_i, \beta_i)$ distribution. The $m$th moment of a random variable $c_i$ following a Beta$(\alpha, \beta)$ distribution is
\begin{align}
    E[c_i^s] = \prod_{j=0}^{s-1}\frac{\alpha + j}{\alpha+\beta+j}
\end{align}
Thus, for $\alpha=(1,...,1)$, $\alpha_i = 1$ and $\beta_i = k-1$, and we obtain
\begin{align}
    E[c_i^s] = \prod_{j=0}^{s-1}\frac{1 + j}{k+j} = \frac{s!}{\frac{(k+s-1)!}{(k-1)!}} = \frac{s!(k-1)!}{(k+s-1)!} = \frac{1}{{k+s-1 \choose s}}
\end{align} 
Specialising for $s=1, 2, 4$
\begin{align}
    E[c_j^4]&=\frac{24}{k(k+1)(k+2)(k+3)}, \label{eq:dirichlet_e4} \\
    E[c_j^2] &= \frac{2}{k(k+1)}, \label{eq:dirichlet_e2} \\
    E[c_j] &= \frac{1}{k} \label{eq:dirichlet_e}
\end{align}
Further,
\begin{align}
    E\left[\prod_{i=1}^k c_i^{m_i} \right] =\frac{\Gamma \left(\sum \limits _{i=1}^k\alpha_i\right)}{\Gamma \left[\sum \limits _{i=1}^k(\alpha_i+m_i)\right]}\times \prod _{i=1}^k\frac{\Gamma (\alpha_i+m_i)}{\Gamma (\alpha_i)}.
\end{align}
Hence,
\begin{align}
    E[c_i c_j] &= \frac{1}{k(k+1)} \\
    Cov(c_i, c_j) &= E[c_i c_j] - E[c_i]E[c_j] \\ 
    &= \frac{1}{k(k+1)} - \frac{1}{k^2} \\
    &= - \frac{1}{k^2(k+1)}\\
\end{align}
Setting $n=2$, we obtain
\begin{align}
    E[c_i^2 c_j^2] &= \frac{4}{k (k+1) (k+2) (k+3)} \\
    Cov(c_i^2, c_j^2) &= E[c_i^2 c_j^2] - E[c_i^2]E[c_j^2] \\ 
    &= \frac{4}{k (k+1)}\left ( \frac{1}{(k+2)(k+3)} - \frac{1}{k(k+1)} \right ) \\
    &= - \frac{8(2k+3)}{k^2 (k+1)^2 (k+2) (k+3)} \\
\end{align}

\section{Brief review of essential concepts}

\subsection{Lie groups}
\label{sec:lie_groups}

The Lie group $\mathcal{G}_j$ generated by a parametrised quantum circuit takes the form~\citep{DAlessandro2021}
\begin{align}
    \mathfrak{g}_j &= \left \langle  \left \{\partial_{\theta_i^{(j)}} U_j(\vec{\theta^{(j)}}) : i \in \{1,...,|\vec{\theta^{(j)}}|\} \right \} \right \rangle  \\
    \mathcal{G}_j &= e^{\mathfrak{g}_j} = \{e^H: H \in \mathfrak{g_j}\},
\end{align}
where $|\vec{\theta^{(j)}}|$ is the length of the vector $\theta^{(j)}$ and $\langle S \rangle$ denotes the \emph{Lie algebra} generated by a set $S$ i.e. the set $S^\infty$, where $S^{(i)} = \{[H_1, H_2] : H_1, H_2 \in S^{(i-1)}\}, S^{(0)} = S$.

\subsection{Free fermion systems}
\label{sec:free_fermions}

Free fermion systems are idealised systems of weakly interacting electrons. Their particularly simple form is useful in applications such as the approximate description of materials via density functional theory~\citep{Jones2015}, or in the Hartree-Fock approximation to interacting quantum states~\citep{Cao2019}. Moreover, these systems feature as emerging quasiparticles in the BCS theory of superconductivity~\citep{Bardeen1957}. Separately, computations based on these systems have been shown to be equivalent to the problem of finding perfect matchings on graphs~\citep{Valiant}, showing that they can be simulated classically in polynomial time and giving rise to the term \emph{matchgate}. In this section, we give the notation and main concepts we use here in the context of free fermion systems. 

A system of $n$ fermions is described by a set of creation and annihilation operators: $a_j^{\dagger}, a_j$ satisfying the Canonical Anticommutation Relations (CAR),
\begin{align}
    &\{a_j, a_k^{\dagger}\} = \delta_{jk}\mathbb{I},\\
    &\{a_j, a_k\} = 0,
\end{align}
where $\{a, b\} = ab + ba$. The creation operators define basis states for multi-fermion systems as
\begin{align}
    \ket{x_1 \dots x_n} := (a_1^{\dagger})^{x_1} \dots (a_n^{\dagger})^{x_n}\ket{0}.
\end{align}
For instance, $\ket{1011} := a_1^{\dagger}a_2^{\dagger} a_4^{\dagger}\ket{0}$. It is common to use an equivalent operator representation based on \textit{Majorana operators}, which are defined as
\begin{align}
    c_{2j-1} &= a_j + a_j^{\dagger},\\
    c_{2j} &= i(a_j - a_j^{\dagger}),
\end{align}
which obey the relations
\begin{align}
    \{c_j, c_k\} &= 2\delta_{jk},\\
    c_j^2 &= 2I.
\end{align}
Free fermion systems can be described in terms of a quadratic fermionic Hamiltonian (which we abbreviate to ``quadratic Hamiltonian"). In terms of Majorana fermions, such a Hamiltonian takes the form
\begin{align}
\label{eq:quadratic_hamiltonian}
  H = i \sum_{j, k} h_{j, k} c_j c_k.  
\end{align}
A fermionic Gaussian unitary is obtained by exponentiating a quadratic Hamiltonian; these correspond to elements of the Lie group $SO(2N)$.

\section{Moments of the Haar measure over the special orthogonal group \citep{diaz2023}}\label{sec:rank1exps}
Here, we summarise the results of \cite{diaz2023}, who provide closed form expressions of $E[m_{i}]$, $E[m_{i}^2]$ when the Lie group is $\mathcal{G} = SO(2N)$, and use this to compute these moments for quadratic observables and fermionic Gaussian initial states. We know from~\cite[Eq. (54)]{diaz2023} that
\begin{align}
    E[m_j] = \frac{1}{2^N}(\tr(\rho_0)\tr(O) + \tr(\rho_0 P)\tr(O P)). \label{eq:first_moment}
\end{align}
where $P := Z_1...Z_n := (-i)^{n} c_1...c_{2n}$ is the parity operator. Define
\begin{align}
    c^b := c_1^{\alpha_1}...c_{2N}^{\alpha_{2N}}, \quad b = (\alpha_1, ... \alpha_{2N}), \quad \alpha \in \{0, 1\}
\end{align}
Assume the observable is a string of Majoranas. Then, the only case where $O$ is proportional to the identity is when $b=(0,...,0)$; in all other cases, it is a product of Majoranas. Likewise, the only case where $OP$ is proportional to the identity is when $b=(1,...,1)$; in all other cases, it will be a product of Majoranas. In particular, note that the expectation value \eqref{eq:first_moment} will vanish for any quadratic observable regardless of the choice of initial state. 

From \citep[Eq. (5)]{diaz2023}, we also know that
\begin{align}
    \label{eq:diaz_em2_gaussian}
    E[m_j^2] &= \sum_{\kappa = 0}^{2N} \frac{P_\kappa(\rho) P_\kappa(O) + C_\kappa(\rho) C_\kappa(O)}{{2N \choose \kappa}},
\end{align}
where
\begin{align}
    P_\kappa\left(\sum_b a_b c^b \right) &:= (-1)^{\left\lfloor\frac{\kappa}{2}\right\rfloor} \sum_{|b|=\kappa} a_b^2,\\
    C_\kappa\left(\sum_b a_b c^b\right) &:= i^{\kappa \, \text{mod} \, 2} \sum_{|b|=\kappa} (-1)^{\sum_j \alpha_j(j  - 1)}\overline{a_b}a_{\bar{b}}.
\end{align}
with 
\begin{align}
\bar b &= (1- \alpha_1, ... 1-\alpha_{2N}), \\
|b|&=\sum_i \alpha_i.
\end{align}
For a fermionic Gaussian state $\rho$ (see \citep[Sec. VIII A.]{diaz2023})
\begin{align}
    \label{eq:diaz_em2}
    P_\kappa(\rho) &:= {N \choose \kappa/2}\\
    C_\kappa(\rho) &:= P_\kappa(\rho).
\end{align}
Thus, the desired expression is (see \citep[Eq. (97)]{diaz2023})
\begin{align}
    E[m_j^2] &= \sum_{\kappa' = 0}^{n} \frac{1}{2^N} {N \choose \kappa'} {2N \choose 2\kappa'}^{-1}  (P_{2\kappa'}(O) + C_{2\kappa'}(O)),
\end{align}
where $\kappa' = \kappa / 2$. Since, for a quadratic observable,
\begin{align}
    P_2(O) &= \tr(O^2), \\
    P_\kappa(O) &= 0, \text{ for } \kappa \neq 2, \\ 
    C_\kappa(O) &= 0, 
\end{align}
we obtain
\begin{align}
    E[m_j^2] = \frac{\tr(O^2)}{2^N}\frac{1}{2N - 1}.
\end{align}

\section{Derivations of results for specific Lie groups}

\subsection{Linear combinations of fermionic Gaussian unitaries / matchgate circuits}
\label{sec:app_lcu_ff}

In this section, we give a version of \cref{thm:var_dirichhaar} specialised to the case where we have a linear combination of parametrised fermionic Gaussian unitaries.

\begin{lemma*}{\ref{lem:ff_abssq}}
\label{lem:ff_abssq_app}
If the Lie group is the space of all fermionic Gaussian unitaries on $N$ qubits i.e.\ $\mathcal{G} = SO(2N)$, and $O$ is an observable, then
\begin{align}
    \label{eq:ff_abssq_app}
    E[|m_{ij}|^2] &= \frac{1}{2^{2N}} \tr(\rho_0^2)\tr(O^2) + \frac{2}{2^{2N}} \tr(\rho_0^2P)\tr(O^2P) + \frac{1}{2^{2N}} \tr(\rho_0P\rho_0P)\tr(OPOP)
\end{align}
where $P$ is the fermionic parity operator.
\end{lemma*}
\begin{proof}
\label{proof:ff_abssq}
    \begin{align}
        E[|m_{ij}|^2] =& E[\tr(\rho_0U_i^\dagger O U_j)\tr(\rho_0U_j^\dagger O U_i)]\\
        =& E[\tr((\rho_0 \otimes \rho_0) (U_i^\dagger \otimes U_i^T) (O \otimes O) (U_j \otimes U_j^*))]\\
        =& \tr((\rho_0 \otimes \rho_0) E[U_i^\dagger \otimes U_i^T] (O \otimes O) E[U_j \otimes U_j^*])\\
        =& \tr((\rho_0 \otimes \rho_0) E[U_i \otimes U_i^*] (O \otimes O) E[U_j \otimes U_j^*])
    \end{align}
    We may now use the identity $E_{U \sim Haar(\mathcal{G})}[U \otimes U^*] = \sum_i^{dim(commutant)} \ket{P_i}\rangle\langle\bra{P_i}$, where $\{P_i\}_{i}$ forms an orthonormal basis for the commutant of $\mathcal{G}$. We refer the reader to \cite{mele2024introduction} for further discussion on expectations calculated over Haar measures.
    For $\mathcal{G} = SO(2N)$, an orthonormal basis for the commutant consists of $\{\frac{1}{\sqrt{2^n}}I, \frac{1}{\sqrt{2^n}}P \}$, where $P$ is the fermionic parity operator. Using this expression, we have
    \begin{align}
        E[|m_{ij}|^2] =& \frac{1}{2^{2N}} \tr((\rho_0 \otimes \rho_0) (\ket{I}\rangle\langle\bra{I} + \ket{P}\rangle\langle\bra{P}) (O \otimes O) (\ket{I}\rangle\langle\bra{I} + \ket{P}\rangle\langle\bra{P}))\\
        =& \frac{1}{2^{2N}} \tr(\rho_0^2)\tr(O^2) + \frac{2}{2^{2N}} \tr(\rho_0^2 P)tr(O^2P) + \frac{1}{2^{2N}} \tr(\rho_0 P\rho_0 P)\tr(OPOP)
    \end{align}
    which is the desired expression.
    This expression may be expressed in terms of the commuting and anti-commuting components of the initial state and measurement operator. We have for any linear operator $A$ on $\mathbb{C}^d$, $A = A_+ + A_-$, where $A_+ := \Pi_+A\Pi_+ + \Pi_-A\Pi_-$, $A_- := \Pi_+A\Pi_- + \Pi_-A\Pi_+$ are the parity preserving and parity flipping components of $A$. Here $\Pi_+$ and $\Pi_-$ are projectors onto the even and odd subspaces, respectively.
    Employing this decomposition for $O$ and $\rho_0$, we may rewrite our previous expression as
    \begin{align}
    E[|m_{ij}|^2] &= \frac{1}{2^{2N}}\tr(\rho_0^2)\tr(O^2) + \frac{2}{2^{2n}}\tr(\rho_{0+}^2P)\tr(O_+^2P) + \frac{1}{2^{2N}}\tr(\rho_{0+}^2 - \rho_{0-}^2) 
    \end{align}
    \qed
\end{proof}

\begin{corollary*}{\ref{cor:ffvar}}
The expectation value of a quadratic observable $O$ for an LCU of Haar random fermionic Gaussian unitaries with an even, pure fermionic Gaussian initial state, has variance 
\begin{align}
\label{eq:ffvar_app}
    \text{\textit{Var}}[m] = \tr(O^2/2^N) \left [ \frac{24}{(k+1) \cdot (k+2) \cdot (k+3) \cdot (2N-1)} + \frac{4 \cdot (k-1)}{(k+1) \cdot (k+2) \cdot (k+3) \cdot 2^{N-2}} \right ].
\end{align}
\end{corollary*}
\begin{proof}
\label{proof:ffvar}
For such a system, using the results from~\cite{diaz2023} in \cref{sec:rank1exps}, the moments can be worked out to be
\begin{align}
    E[m_i] =&~0,\\
    E[m_i^2] =&~\frac{\tr(O^2/2^N)}{2N-1}.
\end{align}
For the $E[|m_{ij}|^2]$ term, we use \cref{lem:abssquare}. Since $O$ is assumed to be quadratic, then  $PO = OP$. Moreover, since $\rho_0$ is assumed to be pure and even, $P\rho_0 = \rho_0 P = \rho_0$, $\tr(\rho_0^2) = 1$. Thus,  we have
\begin{align}
    E[|m_{ij}|^2] &= \frac{1}{2^{2N}} \tr(\rho_0^2)\tr(O^2) + \frac{2}{2^{2N}} \tr(\rho_0^2P)\tr(O^2P) + \frac{1}{2^{2N}} \tr(\rho_0P\rho_0P)\tr(OPOP)\\
    &= \frac{1}{2^{2N}} \tr(O^2) + \frac{1}{2^{2N}} \tr(O^2) \\
    &= \frac{\tr(O^2)}{2^{2N-1}}
\end{align}

Plugging these values into \eqref{eq:dirichlet_lcu_var_app}, we get \eqref{eq:ffvar_app}.
\qed
\end{proof}
By observation this implies that, for such systems,
\begin{align}
    \text{\textit{Var}}[m] \in O \left (\frac{1}{N \cdot k^3} \right ),
\end{align}
which is the statement of Corollary \ref{cor:bigoZ1}.

For the more general case with no additional assumptions on the observable or the fermionic Gaussian initial state, we get
\begin{corollary}
\label{cor:general_ffvar_app}
The expectation value of an observable $O$ for an LCU of Haar random fermionic Gaussian unitaries with a fermionic Gaussian initial state has variance 
\begin{align}
\label{eq:general_ff_dirichlet_lcu_var}
\var[m] &= \frac{24}{(k+1) \cdot (k+2) \cdot (k+3)}  \cdot \sum_{\kappa' = 0}^N \frac{1}{d} {N \choose \kappa'} {2N \choose 2\kappa'}^{-1}  (P_{2\kappa'}(O) + C_{2\kappa'}(O)) \nonumber \\ &+ \frac{4 \cdot (k-1)}{(k+1) \cdot (k+2) \cdot (k+3)} \cdot \left ( \frac{1}{2^{2N}} \tr(\rho_0^2)\tr(O^2) \,\,\, + \right.\nonumber\\ 
 &+ \left. \frac{2}{2^{2N}} \tr(\rho_0^2 P)tr(O^2P) + \frac{1}{2^{2N}} \tr(\rho_0 P\rho_0 P)\tr(OPOP) \right ) \nonumber \\ 
&-  \frac{4(5k-7)}{(k+1)^2(k+2)(k+3)} \cdot \left ( \frac{1}{2^N}(\tr(\rho_0)\tr(O) + \tr(\rho_0 P)\tr(O P)) \right )^2.
\end{align}
\end{corollary}
\begin{proof}
    The result follows from substituting \eqref{eq:first_moment}, \eqref{eq:diaz_em2_gaussian} from \cref{sec:rank1exps} and \eqref{eq:ff_abssq_app} from \cref{lem:ff_abssq} into \eqref{eq:dirichlet_lcu_var_app} from \cref{thm:var_dirichhaar}.
    \qed
\end{proof}

\begin{corollary*}{\ref{cor:bigoZ1}}
The variance of the expectation value of an observable for an LCU of Haar random fermionic Gaussian unitaries scales as
\begin{align}
    \var[m] \in \Omega \left (\frac{1}{N^s \cdot k^3} \right ).
\end{align}
where $m$ depends on the observable. For quadratic observables, $s=1$.
\end{corollary*}
\begin{proof}
\label{proof:bigoZ1}
    In \cref{cor:general_ffvar_app}, the dominant term is ${N \choose \kappa'} {2N \choose 2\kappa'}^{-1}$, which scales as $\Omega \left(\frac{1}{N^s}\right)$, where $s$ is the largest value of $\kappa'$ that satisfies $P_{2\kappa'}(O) + C_{2\kappa'}(O) \neq 0$. The case for quadratic observables follows directly from \cref{cor:ffvar}.
\qed
\end{proof}

\subsection{Linear combinations of expressive unitaries}
\label{sec:app_lcu_expressive}

\begin{lemma*}{\ref{lem:abssquare}}
    \label{lem:abssquare_app}
    For an observable $O$, and Haar distributed unitaries over a Hilbert space of dimension $d$,
    \begin{align}
        \label{eq:abssq_expressive_app}
        E[|m_{ij}|^2] = \frac{\tr(\rho_0^2)\tr(O^2)}{d^2}.
    \end{align}
\end{lemma*}
\begin{proof}
    \label{proof:abssquare}
    This proof may be derived by applying a result from the Weingarten calculus.
    \begin{align}
        E[|m_{ij}|^2] =& E[\tr(U_i \rho_0 U_j^\dagger O)\tr(U_j \rho_0 U_i^\dagger O)]\\
        =& E[\tr((U_i \otimes U_i^*) (\rho_0 \otimes \rho_0^T) (U_j^\dagger \otimes U_j^T) (O \otimes O))\\
        =& E[\tr((U_i \otimes U_i^*) (\rho_0 \otimes \rho_0^T) (U_j \otimes U_j^*)^\dagger (O \otimes O))
    \end{align}
    Utilising two applications e.g. \citet[Eq. (123)]{mele2024introduction}, this reduces to
    \begin{align}
        E[|m_{ij}|^2] =& \frac{1}{d^2} \tr(O^2) \tr(\rho_0^2)
    \end{align}
\qed
\end{proof}

\begin{corollary*}{\ref{cor:expressive_var}}
\label{cor:expressive_var_app}
The expectation value of an observable $O$ for an LCU of Haar random unitaries has variance 
\begin{align}
\label{eq:expressive_var_app}
   \var[m] &= \frac{24 \cdot (\tr(O^2) - tr(O)^2/2^N) \cdot (\tr(\rho_0^2) - 1/2^N)}{(4^N -1) \cdot (k+1) \cdot (k+2) \cdot (k+3)} \\ &+ \frac{4 \cdot (k-1) \cdot (\tr(O^2)/ 2^N\tr(\rho_0^2) + \tr(O)^2)}{(k+1) \cdot (k+2) \cdot (k+3) \cdot 2^N} \\ &- \frac{\tr(O)^2}{2^N}.
\end{align}
\end{corollary*}
\begin{proof}
\label{proof:expressive_var}
From e.g. \cite[Eq. (278)]{mele2024introduction},
\begin{align}
    \label{eq:expressive_var_e2m}
    E[m_i] = tr(O) / 2^N,
\end{align}
and from e.g. \citep[Eq. (14)]{Ragone2024} 
\begin{align}
    \label{eq:expressive_var_em2}
    E[m_i^2] =\frac{1}{4^N-1} \left ( \tr(O^2) - \tr(O)^2/2^N \right ) \left (\tr (\rho^2) - 1/2^N \right).
\end{align}
Substituting both \eqref{eq:expressive_var_e2m}, \eqref{eq:expressive_var_em2} above, plus \eqref{eq:abssq_expressive_app} from \cref{lem:abssquare}, into    \eqref{eq:dirichlet_lcu_var_app} from \cref{thm:var_dirichhaar}, yields the desired result.
\qed
\end{proof}

\end{document}